\RequirePackage{amsmath,amssymb}
\documentclass[a4paper]{article}
\usepackage{geometry}
\usepackage[T1]{fontenc}
\usepackage{graphicx}
\usepackage{amsmath}
\usepackage{amssymb}
\usepackage{amsthm}
\usepackage{amsfonts}
\usepackage{color}
\usepackage{multirow}
\usepackage{makecell}
\usepackage{textgreek}
\usepackage{hyperref}
\usepackage{array}
\usepackage{authblk}

\newtheorem{lemma}{Lemma}
\newtheorem{theorem}{Theorem}
\theoremstyle{definition}
\newtheorem{definition}{Definition}

\usepackage[linesnumbered,ruled,vlined,noresetcount]{algorithm2e}
\DontPrintSemicolon
\SetKwProg{Fn}{function}{:}{end}
\SetKwProg{when}{when }{ \textbf{do}}{done}
\SetKwProg{For}{for }{ \textbf{do}}{done}
\SetKwProg{Forall}{forall }{ \textbf{do}}{endfor}
\let\oldnl\nl
\newcommand{\nonl}{\renewcommand{\nl}{\let\nl\oldnl}}

\newcommand{\iinput}{\mathtt{input}}
\newcommand{\Call}{{\mathcal{C}_{all}}}
\newcommand{\Red}{\mathtt{A}}
\newcommand{\Blue}{\mathtt{B}}
\newcommand{\tie}{\mathtt{T}}

\newcommand{\Rank}{\mathtt{rank}}
\newcommand{\Role}{\mathtt{role}}

\newcommand{\leader}{\mathtt{leader}}
\newcommand{\resetcount}{\mathtt{resetcount}}

\newcommand{\Resetting}{\mathtt{Resetting}}
\newcommand{\Settled}{\mathtt{Settled}}
\newcommand{\Unsettled}{\mathtt{Unsettled}}
\newcommand{\timer}{\mathtt{timer}}
\newcommand{\ans}{\mathtt{answer}}
\newcommand{\Pem}{{\mathcal{P}_{EM}}}

\title{Time- and Space-Optimal Silent Self-Stabilizing Exact Majority in Population Protocols}
\date{}

\author[1]{Haruki Kanaya}
\affil[1]{Nara Institute of Science and Technology, Nara, Japan}
\author[1]{Ryota Eguchi}
\author[1]{Taisho Sasada}
\author[2]{Fukuhito Ooshita}
\affil[2]{Fukui University of Technology, Fukui, Japan}
\author[1]{Michiko Inoue}

\begin{document}

\maketitle

\sloppy

\begin{abstract}
We address the self-stabilizing exact majority problem in the population protocol model, introduced by Angluin, Aspnes, Diamadi, Fischer, and Peralta (2004).
In this model, there are \( n \) state machines, called agents, which form a network. 
At each time step, only two agents interact with each other, and update their states.
In the self-stabilizing exact majority problem, each agent has a fixed opinion, \( \Red \) or \( \Blue \), and stabilizes to a safe configuration in which all agents output the majority opinion from any initial configuration.

In this paper, we show the impossibility of solving the self-stabilizing exact majority problem without knowledge of \( n \) in any protocol.  
We propose a silent self-stabilizing exact majority protocol, which stabilizes within \( O(n) \) parallel time in expectation and within \( O(n \log n) \) parallel time with high probability, using \( O(n) \) states, with knowledge of $n$.  
Here, a silent protocol means that, after stabilization, the state of each agent does not change.  
We establish lower bounds, proving that any silent protocol requires \( \Omega(n) \) states, \( \Omega(n) \) parallel time in expectation, and \( \Omega(n \log n) \) parallel time with high probability to stabilize.  
Thus, the proposed protocol is time- and space-optimal.
\end{abstract}

\section{Introduction}
The population protocol model, introduced by Angluin, Aspnes, Diamadi, Fischer, and Peralta~\cite{Angluin2006}, is a model of a passively mobile sensor network.
In this model, there are $n$ state machines (called \emph{agents}), and these agents form a network (called the \emph{population}).
Each agent lacks a unique identifier and updates its state through pairwise communication (called \emph{interaction}).
At each time step, only one pair of agents interacts, chosen by a uniform random scheduler.
In this paper, we assume that the network is a complete graph, meaning that each agent can interact with all other agents.
The time complexity is measured in \emph{parallel time}, which is defined as the number of interactions divided by \( n \).

The majority problem requires each agent to have a fixed opinion, either \( \Red \) or \( \Blue \), and to determine the majority opinion.
The agents output the majority opinion, \( \Red \) or \( \Blue \), if there is no tie (i.e., the numbers of agents with input \( \Red \) and agents with input \( \Blue \) are not equal); otherwise, they output \( \tie \).
There are two types of majority problems: approximate and exact.
The approximate majority problem~\cite{Angluin2008,Condon2020} allows a small error; that is, a protocol may stabilize to an incorrect answer with low probability.
On the other hand, the exact majority problem does not allow errors; that is, a protocol stabilizes to a correct answer with probability 1.
The exact majority problem with designated initial states has been widely studied~\cite{AAG18,AGRV15,AAE06,BBBEHKK22,BKKP20,BEFKKR18,Berenbrink2021,BCER17,STDtimespaceop,DV2010,KU18,MAABS15,MAS16,PVV2009}, and a time- and space-optimal protocol~\cite{STDtimespaceop}, which converges within \( O(\log n) \) parallel time and uses \( O(\log n) \) states, has been proposed.

We address the self-stabilizing exact majority problem, which requires each agent to have a fixed opinion, either \( \Red \) or \( \Blue \), and to determine the exact majority opinion starting from any configuration.  
In prior research, a study~\cite{LSMajority} uses loose-stabilization~\cite{SUDO2012100}, which relaxes the closure property of stabilization.  
However, to the best of our knowledge, no study has solved the self-stabilizing exact majority problem.  
We solve this problem with initial knowledge of \( n \).

With initial knowledge of \( n \), there have been some studies on another problem: the self-stabilizing leader election~\cite{BCCDNSX21,CIW12,GGS25}, in which agents elect a unique leader.
In~\cite{BCCDNSX21}, Burman, Chen, Chen, Doty, Nowak, Severson, and Xu solve the self-stabilizing ranking problem, in which agents are assigned unique ranks from \([1, n]\).
They proposed two protocols: a silent protocol, which is both time- and space-optimal ($O(n)$ parallel time, $O(n)$ states), and a non-silent protocol, which is time optimal ($O(\log n)$ parallel time) but requires $\exp(O(n^{\log n} \cdot \log n))$ states.
Here, a silent protocol refers to a protocol in which, after stabilization, the state of each agent does not change.

In the non-silent self-stabilizing ranking protocol~\cite{BCCDNSX21}, agents are ranked by assigning unique names in \( [1, n^3] \). Upon stabilization, each agent holds a list of all agents' names and determines its rank in \( [1, n] \) from the name list.
This protocol works correctly even if the name space is reduced to half its size, while still maintaining the same time complexity and number of states, since Lemma 5.1 of~\cite{BCCDNSX21} still holds, as the name space is \( O(n^3) \).
Thus, using this property, the exact majority problem can be solved by assigning agents with input \( \Red \) a name space in \( [1, {n^3}/2] \) and agents with input \( \Blue \) a name space in \( [{n^3}/2 + 1, n^3] \).  
Also, each agent can determine its rank in \([1, n]\) from the name list.
This protocol is time-optimal since, in self-stabilization, all agents must interact at least once, and by the coupon collector's problem, \( \Omega(n\log{n}) \) interactions (which corresponds to \( \Omega(\log{n}) \) parallel time) are required.

\subsection{Our Contribution}

In this paper, we address a silent protocol that solves the self-stabilizing exact majority problem.
Our contribution is summarized in Table~\ref{tab:result}.
First, we show the impossibility that any protocol without knowledge of  $n$ cannot solve the self-stabilizing exact majority.
Second, we also show the lower bounds: any silent protocol with knowledge of \( n \) requires at least \( n\) states, and any silent protocol requires at least \( \Omega(n) \) parallel time in expectation and \( \Omega(n \log n) \) parallel time with high probability to stabilize.
Here, the phrase ``with high probability'' refers to with a probability of \( 1 - O(1/n) \).
Finally, we propose a silent protocol, $\Pem$, which stabilizes to a silent configuration \footnote{
A silent protocol stabilizes to a silent configuration; Thus, its time complexity is measured by silence time, where silence time is defined as parallel time until reaching a silent configuration.  
However, comparing the lower bound of stabilization time with the upper bound of silence time is not an issue, since stabilization time is no more than silence time.
} within $O(n)$ parallel time in expectation and within $O(n\log{n})$ parallel time with high probability using $O(n)$ states, with knowledge of \( n \).

\begin{table}[htb]
    \centering
    \caption{Overview of our results. The variable $n$ denotes the number of agents. W.H.P.\ means with high probability.}
    \label{tab:result}
    \begin{tabular}{|c|c|c|c|c|}
        \hline
         & knowledge & states & expected parallel time & parallel time W.H.P. \\
        \hline
        Impossibility & without $n$ & - & - & -\\
        Lower Bound (space) & $n$ & $n$ & - & -\\
        Lower Bound (time) & $n$ & - & $\Omega(n)$ & $\Omega(n\log{n})$\\
        Protocol & $n$ & $O(n)$ & $O(n)$ & $O(n\log{n})$\\
        \hline
    \end{tabular}
\end{table}

\subsection{Organization of This Paper}
Section~\ref{sec:preliminaries} defines the model and problem. Section~\ref{sec:tools} introduces existing protocols used in our protocol.
In Section~\ref{sec:IMP}, we show that any protocol without knowledge of \( n \) cannot solve the self-stabilizing exact majority problem.  
Section~\ref{sec:Bound} presents lower bounds for any silent protocol that solves the self-stabilizing exact majority problem and proposes a protocol that matches these lower bounds.
We conclude and discuss future directions in Section~\ref{sec:conclude}.

\section{Preliminaries}\label{sec:preliminaries}
We denote the set of positive integers by $\mathbb{N}$, and the set of non-negative integers by $\mathbb{N}_{0}$.

\subsection{Population Protocols}
A \emph{population} is represented by a bidirectional complete graph $G=(V,E)$, where $V$ denotes the set of agents, and $E=\{(u,v)\in V\times V \mid u\ne v\}$ denotes the set of ordered pairs of agents that can interact, and let $n=|V|$.
Since \( G \) is a bidirectional complete graph in this paper, we simply represent the population as \( V \), the set of agents.
A protocol is defined as a 5-tuple $\mathcal{P}=(Q,X,Y,\delta,\pi_{out})$, where $Q$ is the set of agents states, $X$ is the set of input symbols, $Y$ is the set of output symbols, $\delta:(Q\times X)\times (Q\times X) \to Q\times Q$ is the state transition function, and $\pi_{out}:Q\times X\to Y$ is the output function.
The variable \( \mathtt{var} \) of agent \( u \) is denoted as \( u.\mathtt{var} \).  
Similarly, the input of each agent \( u \) is denoted by \( u.\iinput \).
Each agent outputs \( \pi_{out}(s, x) \in Y \) at each time step when it is in state \( s \in Q \) and its input is \( x \in X \).
In each interaction, the interacting agents \( u \) and \( v \) are assigned the roles of initiator and responder, respectively, breaking the symmetry of the interaction.
Suppose that two agents $u,v$ interact, with $u$ as the initiator and $v$ as the responder, while they are in states $p$ and $q$, respectively, and their inputs are $x$ and $y$ respectively.
They then update their states to $p'$ and $q'$, respectively, where $\delta((p,x),(q,y))=(p',q')$.

A \emph{configuration} $C:V\to Q\times X$ represents the state and input of all agents.
A configuration $C$ changes to $C'$ via an interaction $(u,v)\in E$, denoted by $C\xrightarrow{(u,v)}C'$ if $(C'(u),C'(v))=\delta(C(u),C(v))$ and $\forall w\in V\setminus \{u,v\}:C'(w)=C(w)$.
When a configuration $C$ changes to $C'$ via an interaction, we say that $C$ can change to $C'$ denoted by $C\rightarrow C'$.
The set of all configurations by a protocol $\mathcal{P}$ is denoted by $\Call(\mathcal{P})$.
A uniform random scheduler \( \Gamma = \Gamma_0, \Gamma_1, \dots \) determines which ordered pair of agents interacts at each step, where \( \Gamma_t \in E \) (for \( t \geq 0 \)) is a random variable satisfying  \( \forall (u,v) \in E, \forall t \in \mathbb{N}_{0}: \Pr(\Gamma_t = (u,v)) = \frac{1}{n(n-1)} \).
Similarly, a deterministic scheduler $\gamma=\gamma_0,\gamma_1,\dots$ determines which ordered pair of agents interacts at each time step, where $\gamma_i\in E$ (for $i\ge 0$) is a predetermined value.
Note that $\gamma$ can be a possible value of $\Gamma$ with positive probability.

An \emph{execution} of \( \mathcal{P} \) is defined as an infinite sequence of configurations \( \Xi_{\mathcal{P}}(C_0, g)\\ = C_0, C_1, \dots \) that starts from an initial configuration \( C_0 \in \Call(\mathcal{P}) \), and satisfies \( \forall i \in \mathbb{N} : C_{i-1} \to C_i \), where $g$ is a uniform random or deterministic scheduler.
Note that if \( g \) is a uniform random scheduler, then \( C_i \) is a random variable.
A configuration $C'$ is reachable from $C_0$ if there exists an execution $\Xi_{\mathcal{P}}(C_0, \Gamma)=C_0,C_1,\dots$ such that $\exists i\in\mathbb{N}:C_i=C'$.
In an execution under a uniform random scheduler, if a configuration \( C \) appears infinitely often, then every configuration reachable from \( C \) also appears infinitely often\footnote{
This property holds since a uniform random scheduler is globally fair with probability 1, as shown in~\cite{chatzigiannakis_et_al:DagSemProc.09371.4}.
}. 
An execution is said to reach a set of configurations \( \mathcal{C} \) if it reaches a configuration in \( \mathcal{C} \).
Similarly, an execution is said to belong to a set of configurations \( \mathcal{C} \) if its initial configuration is in \( \mathcal{C} \).
We assume that the input remains unchanged during execution, meaning that \(\forall v \in V ,\forall i \in \mathbb{N} : C_i(v).\iinput = C_0(v).\iinput \) holds for an execution \( \Xi_{\mathcal{P}}(C_0, g) = C_0, C_1, \dots \), where $g$ is a uniform random or deterministic scheduler.
A configuration \( C \) is safe if and only if, for any execution starting from \( C \), the outputs of all agents never change.  
Similarly, a configuration \( C \) is silent if and only if, for any execution starting from \( C \), the states of all agents never change.
A protocol is silent if and only if an execution reaches a silent configuration with probability 1.

The time complexity of a protocol is defined as the number of interactions divided by \( n \).  
This is referred to as either \emph{parallel time} or simply \emph{time}.  
The stabilization time is defined as the time until an execution of the protocol under a uniform random scheduler reaches a safe configuration.  
Similarly, the silence time is defined as the time until an execution of the protocol under a uniform random scheduler reaches a silent configuration.
The phrase \textit{``with high probability''} refers to with a probability of \( 1 - O(1/n) \).

\subsection{Self-Stabilizing Exact Majority}
The self-stabilizing exact majority problem requires that each agent outputs the exact majority opinion from the two types of opinions, $\Red$ and $\Blue$, held by agents, starting from any initial configuration.  
We define \( V_a \) as the set of agents with the opinion \( \Red \), and \( V_b \) as the set of agents with the opinion \( \Blue \).
Note that we assume the input of each agent does not change during execution; thus, \( V_a \) and \( V_b \) remain unchanged throughout execution.  
The exact majority opinion is as follows:  
If \( |V_a| < |V_b| \), the exact majority opinion is \( \Blue \); if \( |V_a| > |V_b| \), it is \( \Red \); otherwise, it is \( \tie \).

\begin{definition}
A protocol $\mathcal{P} = (Q, X, Y, \delta, \pi_{out})$ is a self-stabilizing exact majority protocol if and only if $X = \{\Red, \Blue\}$, $Y = \{\Red, \Blue, \tie\}$, and the following conditions hold for any initial configuration $C_0 \in \Call(\mathcal{P})$: \begin{itemize}
    \item For any safe configuration $C_{\text{safe}}$ reachable from $C_0$, it holds that $\forall v \in V: \pi_{out}(C_{\text{safe}}(v)) = y$, where $y \in Y$ is the exact majority opinion of agents. 
    \item Any execution starting from $C_0$ reaches safe configurations with probability 1.
\end{itemize}
\end{definition}

\section{Tools}\label{sec:tools}
\subsection{Epidemic Protocol}
The epidemic protocol proposed by Angluin, Aspnes, and Eisenstat~\cite{AAE06} is used in \textsc{Propagate-Reset} described later.
In this protocol, each agent has a variable \( x \in \{0,1\} \), and when two agents \( u \) and \( v \) interact, \( u.x \) and \( v.x \) are both updated to \( \max(u.x, v.x) \).
Suppose that all agents have \( x = 0 \).  
Once one agent's \( x \) becomes 1, all agents' \( x \) will become 1 within \( \Theta(\log{n}) \) time in expectation and with high probability.

\subsection{Ranking Protocol}\label{subsec:ranking}
In this subsection, we explain the silent self-stabilizing ranking protocol proposed by Burman et al.~\cite{BCCDNSX21} since we will use it in Section~\ref{subsec:em}.
A self-stabilizing ranking protocol assigns agents unique ranks in \([1, n]\) from any initial configuration, and the agents maintain their ranks forever.
Their protocol is divided into two parts, \textsc{Propagate-Reset} and \textsc{Optimal-Silent-SSR}.
In their protocol, an agent has a $\Role\in \{\Resetting,\Settled, \Unsettled\}$ to reduce the number of states, and they define variables for each role.
The number of states becomes the sum of the states required for each $\Role$.
We briefly explain their protocol, but refer to~\cite{BCCDNSX21} for more details and proofs.

\subsubsection{\textsc{Propagate-Reset}}
In this protocol, when some inconsistency is detected in the population, all agents are reset to a specific state with high probability.
Each agent uses two variables, $\Role$ and $\resetcount$, which represent the agent's role and a timer used to determine whether a reset has been completed.
When agents detect some inconsistency, they set \( \Role \) to \( \Resetting \) and \( \resetcount \) to \( R_{max} (= 60 \log n) \).
The role of $\Resetting$ propagates to all agents through the epidemic protocol, treating \( \Resetting \) as 1 and all other roles as 0.
Eventually, all agents are reset within $O(\log{n})$ time with high probability.
After resetting, all inconsistencies are eliminated.
We have following lemmas from~\cite{BCCDNSX21}.

\begin{lemma}[Lemma 3.2 in \cite{BCCDNSX21}]\label{ssrk:prop}
If some inconsistency is detected, all agents' $\Role$ become $\Resetting$ within $4\log{n}$ time with high probability.  
\end{lemma}

\begin{lemma}[Corollary 3.5 in \cite{BCCDNSX21}]\label{ssrk:reset}
Starting from any configuration, an execution reaches a configuration in which no agent has \( \Role = \Resetting \) within \( O(n) \) time with high probability.
\end{lemma}

\subsubsection{\textsc{Optimal-Silent-SSR}}
In this protocol, agents are assigned unique ranks in $[1,n]$ from any initial configuration.
The functions of \( \Role \) for agents are as follows: \( \Role = \Settled \) means the agent has been assigned a rank, \( \Role = \Unsettled \) means the agent has not been assigned a rank, and \( \Role = \Resetting \) means the agent is executing \textsc{Propagate-Reset}.
When some rank conflict or some inconsistency is detected, \textsc{Propagate-Reset} starts.  
At this point, as mentioned above, \( \Role \) is set to \( \Resetting \), \( \resetcount \) is set to \( R_{max} \), and additionally, \( \leader \) is set to \( L \).
Here, \( \leader \in \{L, F\} \) is a variable of an agent with \( \Role = \Resetting \) in this protocol, indicating whether the agent is a leader.
While resetting, agents elect a unique leader through a process where, when leaders meet, one of them becomes a follower (\(F\)).
After resetting, all agents are unranked, and a unique leader agent is elected.
The leader's \( \Role \) becomes \( \Settled \), with its \( \Rank \) set to \( 1 \).
Each agent with \( \Rank = i \) creates agents with \( \Rank = 2i \) and \( \Rank = 2i + 1 \) if the newly assigned \( \Rank \) does not exceed \( n \).
Eventually, all agents are assigned unique ranks.
Let $s_{rank}$ be the expected silence time of \textsc{Optimal-Silent-SSR}.
We have the following lemma from~\cite{BCCDNSX21}.

\sloppy
\begin{lemma}[Theorem 4.3 in \cite{BCCDNSX21}]\label{ssrk:thm}
\normalfont\textsc{Optimal-Silent-SSR}\itshape~is a silent self-stabilizing ranking protocol with $O(n)$ states and $s_{rank}=O(n)$ expected time.
\end{lemma}
\fussy

\section{Impossibility}\label{sec:IMP}

In this section, we provide the impossibility.
We prove that no protocol can solve the self-stabilizing exact majority problem without knowledge of \( n \), using the same approach as the proof of the impossibility of a self-stabilizing bipartition protocol by Yasumi, Ooshita, Yamaguchi, and Inoue~\cite{YOYI19}.

\begin{theorem}
There is no protocol that solves the self-stabilizing exact majority without knowledge of $n$.
\end{theorem}

\begin{proof}
For contradiction, we assume there is a protocol $\mathcal{P}$ that solves self-stabilizing exact majority without knowledge of $n$.
Consider a population $V$, where $|V|\ge 4$, $|V_a|\ge 2$, $|V_b|\ge 2$, and $|V_a|>|V_b|$.
From the assumption, $\mathcal{P}$ solves the problem for $V$ and $V_b$.
Let \( C_0 \) be any configuration of \( V \), and let \( \gamma=\gamma_0,\gamma_1,\dots \) be a deterministic scheduler on $V$ that makes \( C_0 \) eventually reach a safe configuration.
For an execution \( \Xi_{\mathcal{P}}(C_0, \gamma) = C_0, C_1, \dots \) on \( V \), there exists a safe configuration \( C_t \) in which all agents output \( \Red \).

Similarly, consider another population \( V' \) consisting of all agents in \( V_b \), and let \( C'_0 \) be a configuration of \( V' \) satisfying \( \forall v \in V' : C'_0(v) = C_t(v) \).  
Let \( \gamma' = \gamma'_0, \gamma'_1, \dots \) be a deterministic scheduler on \( V' \) that makes \( C'_0 \) eventually reach a safe configuration.
For an execution \( \Xi_{\mathcal{P}}(C'_0, \gamma') = C'_0, C'_1, \dots \) on \( V' \), there exists a safe configuration \( C'_{t'} \) in which all agents output \( \Blue \).

From these executions, let \( \gamma'' = \gamma''_0, \gamma''_1, \dots \) be a deterministic scheduler on \( V \) satisfying $\forall i\in[0,t-1]:\gamma''_i=\gamma_i \land \forall j\in [0,t'-1]:\gamma''_{t+j}=\gamma'_j$.
Then, let \( C''_0 = C_0 \), and consider the execution \( \Xi_\mathcal{P}(C''_0, \gamma'') \).  
In this execution, \( C''_t \) is a safe configuration in which all agents output \( \Red \).  
However, in \( C''_{t+t'} \), all agents that belong to \( V_b \) output \( \Blue \).  
This contradicts the assumption that \( C''_t \) is a safe configuration.
\end{proof}
Note that a deterministic scheduler can appear as a prefix of a uniform random scheduler with nonzero probability.  
Therefore, this theorem is also applicable to the impossibility result under a uniform random scheduler.

\section{Lower Bound and Upper Bound}\label{sec:Bound}
\subsection{Lower Bounds}

In this subsection, we provide the lower bounds.

\begin{lemma}\label{lb:silent}
Consider a population $V$ where $|V|\ge 4$, $|V_a|\le|V_b|$.
In a silent configuration of a silent self-stabilizing exact majority protocol, the states of agents with input \( \Red \) are all distinct from one another.
\end{lemma}

\begin{proof}
We prove this lemma by contradiction. 
We assume that there exists a silent self-stabilizing exact majority protocol \( \mathcal{P} \) that uses \( |V_a| - 1 \) states for agents with input \( \Red \) in a silent configuration.
Suppose some execution reaches a silent configuration \( C \). From the assumption, there exist two agents with input \( \Red \) that are in the same state \( s_A \) in \( C \). Since \( C \) is a silent configuration, the state satisfies $\delta((s_A,\Red),(s_A,\Red))\to (s_A,s_A)$ and $\pi_{out}(s_A, \Red)\in \{\Blue,\tie\}$.
Next, consider another population with \( n \) agents, where all agents have \( \iinput = \Red \), and an execution starting from the configuration in which all agents' states are \( s_A \).
Such an execution cannot reach a safe configuration since all agents eternally output $\Blue$ or $\tie$, while the exact majority opinion is $\Red$.
This is a contradiction.
\end{proof}

\begin{theorem}\label{lb:states}
A silent self-stabilizing exact majority protocol requires $n$ states.
\end{theorem}

\begin{proof}
Let \( n \) be an even number no less than 4.  
Consider a population \( V \) with \( n \) agents, where \( |V_a| = n/2 - 1 \).  
From Lemma~\ref{lb:silent}, there exist \( n/2 - 1 \) distinct states of agents with \( \iinput = \Red \) that output \( \Blue \).
Next, consider another population \( V' \) with \( n \) agents, where \( |V'_a| = n/2 \).  
From Lemma~\ref{lb:silent}, there exist \( n/2 \) distinct states of agents with \( \iinput = \Red \) that output \( \tie \).  
Finally, consider a population \( V'' \) with \( n \) agents, where \( |V''_a| = n/2 + 1 \).  
It is clear that there exists at least one state of agents with \( \iinput = \Red \) such that they output \( \Red \).  
Thus, the number of states is at least $(n/2 - 1) + (n/2) + 1 = n$.
\end{proof}

We prove a lower bound for the stabilization time of a silent self-stabilizing exact majority protocol in the same way as the proof of the lower bound for the stabilization time of a silent self-stabilizing leader election protocol by Burman et al.~\cite{BCCDNSX21}.

\begin{theorem}\label{lb:time}
Any silent self-stabilizing exact majority protocol requires $\Omega(n)$ time in expectation and $\Omega(n\log{n})$ time with high probability to reach a safe configuration.
\end{theorem}

\begin{proof}
Consider a population $V$ where $|V|\ge 5 \land |V|\bmod 2=1$, $|V_a|= \lfloor n/2\rfloor$, and $|V_b|= \lceil n/2\rceil$.
From Lemma~\ref{lb:silent}, in a silent configuration $C$, the states of agents whose input are $\Red$ are distinct from one another.
Let \( s_A \) be the state of an agent \( w \) whose input is \( \Red \) in $C$.

For an agent $u\in V$ whose input is $\Blue$, we consider another population $V'=V$ and a configuration $C'$ satisfying $\forall v\in V'\setminus\{u\}:C'(v)=C(v)$ and $C'(u)=(s_A,\Red)$.
In $C'$, all agents output $\Blue$, and the exact majority opinion of $V'$ is $\Red$.
Thus, $C'$ is not safe.
Since $C$ is a silent configuration, in $C'$, unless $u$ and $w$ interact with each other, their states will not be changed.
The probability that $u$ and $w$ interact is $\frac{2}{n(n-1)}$.
Thus, the expected number of interactions until the interaction occurs is $n(n-1)/2=\Omega(n^2)$.
Therefore, divided by $n$, the expected time is $\Omega(n)$.

The probability that \( u \) and \( w \) do not interact in $\alpha \frac{n(n-1)}{2}\log_{e}{n}-1$ interactions is\\ $\left(1-\frac{2}{n(n-1)}\right)^{\alpha\frac{n(n-1)}{2}\log_{e}{n}-1}>e^{-\alpha\log_{e}{n}}=n^{-\alpha}$ for some constant value $\alpha > 0$.
Thus, the number of interactions until $u$ and $w$ interact with each other is $\frac{n(n-1)}{2}\log_{e}{n}-1=\Omega(n^2\log{n})$ with probability $1-O(1/n)$ ($\alpha=1$). 
Therefore, divided by $n$, the time with high probability is $\Omega(n\log{n})$.
\end{proof}

\subsection{Matching Upper Bound}\label{subsec:em}

In this subsection, we propose a silent self-stabilizing exact majority protocol, $\Pem$, using $O(n)$ states and stabilizing within $O(n)$ time in expectation and within $O(n\log{n})$ time with high probability, with knowledge of $n$.
In \( \Pem \), each agent has an input and a state.  
The input of an agent \( a \) is \( a.\iinput \in \{\Red, \Blue\} \) (read-only).  
The state of each agent is represented by separate variables.
Since we use the silent self-stabilizing ranking protocol \textsc{Optimal-Silent-SSR}~\cite{BCCDNSX21}, each agent has the variables used in \textsc{Optimal-Silent-SSR}.  
Additionally, each agent \( a \) has two other variables: \( a.\ans \in \{\phi, \tie, \Red, \Blue\} \) and \( a.\timer \in [0, 7(t_{rank} + 4)] \).
Here, \( t_{rank} \) is a constant satisfying \( s_{rank} \leq t_{rank} \cdot n \) and \( t_{rank} = O(1) \), where \( t_{rank} \cdot n \) denotes a sufficient amount of time for \textsc{Optimal-Silent-SSR} to stabilize.
The number of states of \( \Pem \) is \( O(n) \) since the number of states of \( \ans \) and \( \timer \) is \( O(1) \), and the number of states of the variables used in \textsc{Optimal-Silent-SSR} is \( O(n) \).
Each agent $a$ outputs $\tie$ if $a.\ans = \phi$; otherwise, it outputs $a.\ans$.

\begin{algorithm}[!htb]
\caption{Silent Self-Stabilizing Exact Majority Protocol $\Pem$}
\label{protocol:SSEMN}
\SetKwInOut{VariablesAgent}{Agents' variables}
\SetKwInOut{OutFunc}{Output function $\pi_{out}$}
\setcounter{AlgoLine}{0}
\nonl\when{an agent $a_0$ interacts with an agent $a_1$}{
    Execute \textsc{Optimal-Silent-SSR}~\cite{BCCDNSX21}.\;
    \For{$i\in \{0,1\}$}{
        \If{$a_i.\Role \textnormal{ becomes } \Resetting \textnormal{ in this interaction}$}{
            $a_i.\ans\gets \phi$\;
        }\If{$a_i.\Role \textnormal{ becomes } \Settled \textnormal{ in this interaction} \land a_i.\Rank = \lceil n/2\rceil$}{
            $a_i.\timer\gets 7(t_{rank}+4)$\;
        }
    }
    \If{$\exists i\in\{0,1\}:a_0.\mathtt{role}=a_1.\mathtt{role}=\mathtt{Resetting}\wedge a_i.\ans=\phi \wedge a_{1-i}.\ans\ne \phi$}{
        $a_i.\ans\gets a_{1-i}.\ans$\;
    }
    \If{$a_0.\mathtt{role}=a_1.\mathtt{role}=\mathtt{Settled}$}{
        \If{$a_0.\Rank < a_1.\Rank \land a_0.\iinput = \Blue \land a_1.\iinput = \Red$}{
            Swap the states between $a_0$ and $a_1$.
        }
        \uIf{$n\bmod 2 = 0 \land \exists i\in\{0,1\}:a_i.\Rank = a_{1-i}.\Rank -1 = n/2$}{
            \uIf{$a_i.\iinput = a_{1-i}.\iinput$}{
                $(a_i.\ans, a_{1-i}.\ans)\gets (a_i.\iinput, a_i.\iinput)$\;
            }\Else{
                $(a_i.\ans, a_{1-i}.\ans)\gets (\tie, \tie)$\;
            }
        }\ElseIf{$n\bmod 2=1 \land \exists i\in\{0,1\}:a_i.\Rank=\lceil n/2\rceil$}{
            $a_i.\ans\gets a_i.\iinput$\;
        }
        \If{$\exists i\in\{0,1\}:a_i.\Rank=\lceil n/2\rceil$}{
            \lIf{$a_{1-i}.\Rank=n$}{
                $a_i.\timer\gets \max(0,a_i.\timer-1)$
            }
            \If{$a_i.\timer=0 \wedge a_i.\ans\ne a_{1-i}.\ans$}{
                
                $a_{1-i}.\ans\gets a_i.\ans$\;
                \For{$j\in\{0,1\}$}{
                    $(a_j.\mathtt{role},a_j.\leader,a_j.\mathtt{resetcount})\gets(\mathtt{Resetting},L,R_{max})$\;
                }
            }
        }
    }
}
\end{algorithm}

Our protocol $\Pem$, described in Algorithm~\ref{protocol:SSEMN}, can be summarized as follows:
\begin{itemize}
    \item \textbf{Ranking}. Execute \textsc{Optimal-Silent-SSR} (line 1).
    \item \textbf{Swapping}. Swap the values of all variables between interacting agents, except for $\iinput$, if the $\Rank$ of an agent with input $\Blue$ is less than the $\Rank$ of an agent with input $\Red$ (lines 10--11).
    This eventually places agents with input $\Red$ in the lower ranks and agents with input $\Blue$ in the higher ranks.
    Note that even if the states of pairs of agents are swapped, this does not affect the correctness or the silence time of the ranking protocol, due to the symmetry among agents.
    \item \textbf{Decision}. If $n$ is even, the agent with $\Rank = n/2$ decides the exact majority opinion based on the inputs of both itself and the agent with $\Rank = n/2 + 1$ (lines 12--16). Otherwise, the agent with $\Rank = \lceil n/2 \rceil$ decides the exact majority opinion based on its own input (lines 17--18).
    \item \textbf{Propagation}. 
    After completing the above steps, the agent with $\Rank = \lceil n/2\rceil$ checks the opinions of the other agents. If there is an agent with a different opinion, \textsc{Propagate-Reset} is triggered, and the correct opinion is propagated during resetting. (lines 2--8 and 19--24).
\end{itemize}
In what follows, we explain the details of the Decision part and the Propagation part.

In the Decision part, agents decide the exact majority opinion.
We consider the situation where the Ranking part and the Swapping part are completed.
If $n$ is even, the exact majority opinion is determined by the agents with \( \Rank = n/2 \) and \( \Rank = n/2+ 1 \).
If the inputs of both agents are the same, then the agents with that input are in the majority, making that input the exact majority opinion.  
If not, the number of agents with input \( \Red \) and input \( \Blue \) is the same, so the exact majority opinion is \( \tie \).
If $n$ is odd, the exact majority opinion is the input of the agent with $\Rank=\lceil n/2\rceil$ since the number of agents with that input constitutes a majority.

In the Propagation part, the agent with $\Rank=\lceil n/2\rceil$ propagates its $\ans$ to all agents.
However, if the agent with \( \Rank = \lceil n/2 \rceil \) shares its \( \ans \) to every agent by a direct interaction, it would take \( \Theta(n\log{n}) \) time, since it is equivalent to the coupon collector's problem.
To achieve propagation within \( O(n) \) time in expectation, we use the following method.
The agent with $\Rank = \lceil n/2\rceil$ checks the opinions of the other agents (line 22).
It detects that ranking, swapping, and decision are completed when its \( \timer \) reaches 0.
To detect this, the agent with \( \Rank = \lceil n/2 \rceil \) sets its \( \timer \) to \( 7(t_{rank} + 4) \) when it is ranked (lines 5--6) and decreases it by 1 when it interacts with the agent with \( \Rank = n \) (line 20).
If there is an agent with a different opinion, it can be detected within \( O(n) \) time in expectation and triggers \textsc{Propagate-Reset} by setting its \( \Role \) to \( \Resetting \) (lines 21--24).  
Note that the variables \( \leader \) and \( \resetcount \) in the pseudocode are used in the ranking protocol (see Section~\ref{subsec:ranking}).
During \textsc{Propagate-Reset}, that is, while the \( \Role \) of agents is \( \Resetting \), the correct opinion is propagated to all agents.
Specifically, any agent with \( \Role = \Resetting \) other than the agent with \( \Rank = \lceil n/2 \rceil \) sets its \( \ans \) to \( \phi \) when its \( \Role \) becomes \( \Resetting \) (lines 3--4).
Then, the exact majority opinion overwrites \( \phi \) via the epidemic protocol\footnote{
In the epidemic protocol, the exact majority opinion is treated as 1, and \( \phi \) is treated as 0.
} (lines 7--8).
After the succeeding ranking, swapping, and decision, the agent with \( \Rank = \lceil n/2 \rceil \) again decides the correct opinion (lines 12--18).
At this time, all agents have the correct opinion; thus, \textsc{Propagate-Reset} will not occur. 
These process takes $O(n)$ time in expectation.

To summarize the above, $\Pem$ is as shown in Figure~\ref{fig:flow}.
\begin{figure}
    \centering
    \includegraphics[width=1.0\linewidth]{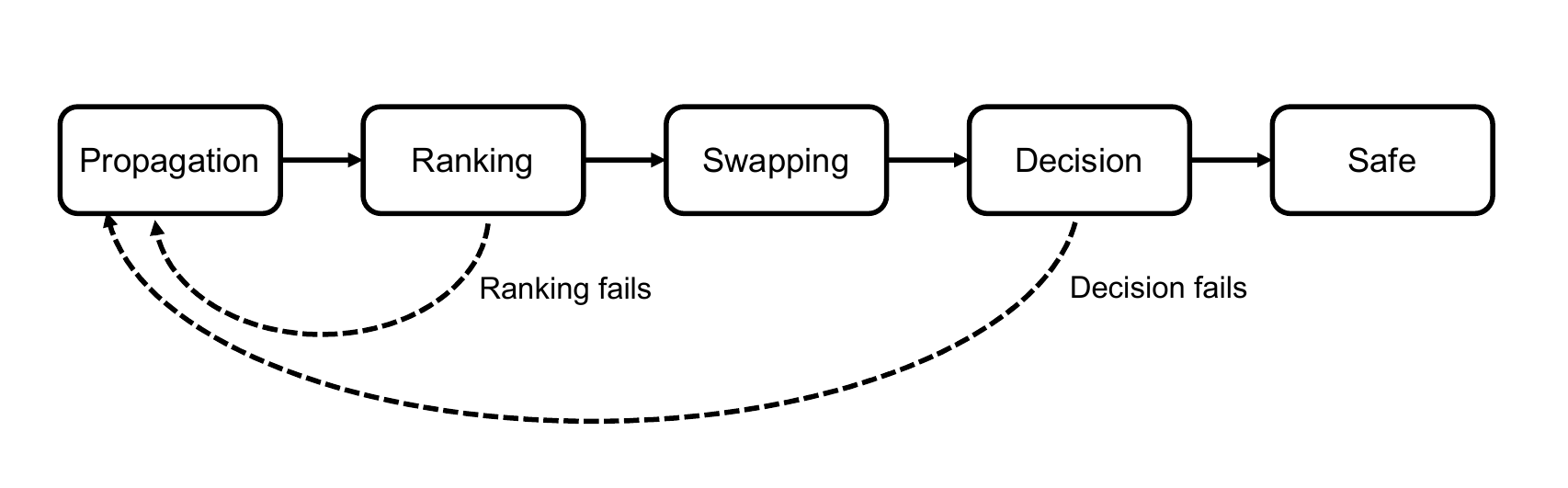}
    \caption{Flow of $\Pem$}
    \label{fig:flow}
\end{figure}
First, agents execute ranking, swapping, and decision.
After that, if $\ans$ of some agent is incorrect, the output of the agent with $\Rank = \lceil n/2 \rceil$ is propagated to all agents by restarting the ranking.
After re-ranking, re-swapping, and re-decision, the outputs of all agents are identical and correct.
Note that these transitions in the figure are probabilistic. 
Thus, if ranking or decision fails, agents also restart the protocol from the beginning by resetting via \textsc{Propagate-Reset}.

\subsubsection{Correctness and Analysis}

\newcommand{\Srank}{\mathcal{S}_{rank}}
\newcommand{\Sswap}{\mathcal{S}_{swap}}
\newcommand{\Tswap}{\mathcal{T}_{swap}}
\newcommand{\NOmid}{\mathcal{NO}_{mid}}
\newcommand{\Sout}{\mathcal{S}_{out}}
\newcommand{\Sdec}{\mathcal{S}_{dec}}
\newcommand{\Stim}{\mathcal{S}_{tim}}
\newcommand{\Sem}{\mathcal{S}_{em}}

In this part, we prove that an execution of \( \Pem \) under a uniform random scheduler reaches a silent configuration within \( O(n) \) time in expectation and $O(n\log{n})$ time with high probability.
To prove it, we define the sets of configurations as follows:
\begin{align*}
\Srank&=\{C\in \Call(\Pem)\mid \forall v\in V:C(v).\Role=\Settled \land \forall u,\forall v\in V:u\ne v\Rightarrow C(u).\Rank\ne C(v).\Rank\},\\
\Sswap&=\{C\in \Srank\mid \forall u,\forall v\in V:C(u).\iinput=\Red \land C(v).\iinput=\Blue \Rightarrow C(u).\Rank<C(v).\Rank\},\\ 
\Tswap&=\{C\in \Srank\mid \exists v\in V:C(v).\Rank=\lceil n/2\rceil\land C(v).\timer\ge 28\},\\
\Sdec&=\{C\in \Sswap \mid \forall v\in V:C(v).\Rank=\lceil n/2\rceil \Rightarrow C(v).\ans \text{ is} \text{ the} \text{ exact} \text{ majority} \text{ opinion}\},\\
\Sout&=\{C\in \Srank\mid \forall v\in V:C(v).\ans \textnormal{ is the exact majority opinion}\},\\
\Stim&=\Sswap \cap \Sout, \textnormal{ and}\\
\Sem&=\{C\in \Stim\mid \forall v\in V:C(v).\Rank=\lceil n/2\rceil \Rightarrow C(v).\timer=0\}.
\end{align*}

\begin{lemma}\label{em:silent}
$\Sem$ is a set of silent configurations.
\end{lemma}

\begin{proof}
In a configuration \( C_0 \in \Sem \), ranking, swapping, and decision are completed, the $\timer$ of the agent with $\Rank=\lceil n/2\rceil$ is $0$, and all agents' \( \ans \) are the same.
Thus, throughout a execution from $C_0$, no agent changes its state.
\end{proof}

\begin{table}[!htb]
    \caption{Stabilization Steps of $\Pem$}
    \centering
    \begin{tabular}{c@{\hspace{15pt}}c@{\hspace{15pt}}c@{\hspace{15pt}}c}
        \hline
        Step & Time & Success Probability & Lemmas\\
        \hline
        $\Call(\Pem) \rightarrow \Srank$ & $O(n)$ & $1/10$ & Lemma~\ref{em:rank}\\
        $\Srank \rightarrow \Tswap\cup\Stim$ & $O(n)$ & $1/20$ & Lemma~\ref{em:tswap}\\
        $\Tswap \rightarrow \Sdec$ & $O(n)$ & $1/8$ & Lemma~\ref{em:dec}\\
        $\Sdec \rightarrow \Stim$ & $O(n)$ & $1/1280$ & Lemma~\ref{em:tim}\\
        $\Stim \rightarrow \Sem$ & $O(n)$ & $1/2$ & Lemma~\ref{em:safe}\\
        \hline
        $\Call(\Pem) \rightarrow \Sem$ & $O(n)$ & $1/4096000$ & Lemma~\ref{em:stab}\\
        \hline
    \end{tabular}
    \label{tab:steps}
\end{table}

We analyze \( \Pem \) by dividing it into several steps, as shown in Table~\ref{tab:steps}.

\begin{lemma}\label{em:rank}
Let $C_0\in \overline{\Srank}$ and $\Xi_\Pem(C_0,\Gamma)$.
The probability that the execution reaches $\Srank$ within $O(n)$ time is at least \(1/10\).
\end{lemma}

\begin{proof}
First, we consider the case that the agent with \( \Rank = \lceil n/2 \rceil \) does not start \textsc{Propagate-Reset} by interacting with an agent that has a different \( \ans \) until the execution reaches $\Srank$.
In this case, from Lemma~\ref{ssrk:thm}, agents complete the ranking within $s_{rank}$ expected time.
From Markov's inequality, the probability that ranking has completed within $2\cdot s_{rank}$ time is at least $1/2$.

Second, we consider the case where the agent with \( \Rank = \lceil n/2 \rceil \) starts \textsc{Propagate-Reset} by interacting with an agent that has a different \( \ans \) until the execution reaches \( \Srank \).
In this case, after \textsc{Propagate-Reset}, an agent will be assigned the rank \( \lceil n/2 \rceil \) within \( 2\cdot s_{rank} \) time with at least \( 1/2 \) probability from Markov's inequality.
At the time, the agent sets \( \timer \) to \( 7(t_{rank} + 4) \).
We analyze the probability that the \( \timer \) does not reach 0 before the ranking is completed.
Since the silence time of the ranking is \( s_{rank} \leq t_{rank} \cdot n \), we analyze the probability that \( \timer \) does not reach 0 during \( t_{rank} \cdot n^2 \) interactions.
Let \( Z \sim B(n^2\cdot t_{rank}, 2/(n(n-1))) \) be a binomial random variable representing the number of times the agent with \( \Rank = \lceil n/2 \rceil \) interacts with the agent with $\Rank=n$ during \( n^2\cdot t_{rank} \) interactions.
From the Chernoff bound (Eq.~(4.2) in~\cite{Arisu}, with \( \delta = 2/3 \)), 
$\Pr(Z\ge(1+2/3)E[Z])\le e^{-2/3(2/3)^2}<4/5$.
Since $(1+2/3)E[Z]\le \frac{20}{3}t_{rank}< 7t_{rank}$, that probability is at last $1/5$.
Thus, the execution reaches $\Srank$ with probability at least $1/10$.
\end{proof}

\begin{lemma}\label{em:swap}
Let $C_0\in \Srank$ and $\Xi_\Pem(C_0,\Gamma)$.
If \normalfont{\textsc{Propagate-Reset}} does not occur, the expected time for the execution to reach \( \Sswap \) is at most \( n \).
\end{lemma}

\begin{proof}
Let \( m \) be the number of agents with input \( \Red \).
Let \( V_B \) be the set of agents with input \( \Blue \) and $\Rank$ at most \( m \), and let \( V_A \) be the set of agents with input \( \Red \) and $\Rank$ greater than \( m \) at initial configuration.
Note that $|V_A|=|V_B|$.
The sizes of \( V_A \) and \( V_B \) each decrease by 1 if and only if an interaction occurs between \( a \in V_A \) and \( b \in V_B \).  
Additionally, in any other interactions, the sizes of \( V_A \) and \( V_B \) remain unchanged.
The probability that an interaction occurs between \( a \in V_A \) and \( b \in V_B \) for each interaction is $\frac{2|V_A||V_B|}{n(n-1)}$.
Thus, the number of interactions until $|V_A|$ and $|V_B|$ become $0$ is $\sum_{i=1}^{|V_A|}\frac{n(n-1)}{2i^2}=\frac{n(n-1)}{2}\sum_{i=1}^{|V_A|}\frac{1}{i^2}<\frac{n(n-1)}{2}\cdot \frac{\pi^2}{6}<n^2$.
Therefore, after dividing by \( n \), the expected time for the execution to reach \( \Sswap \) is at most \( n \).
\end{proof}

\begin{lemma}\label{em:tswap}
Let $C_0\in \Srank \cap \overline{\Tswap}$ and $\Xi_\Pem(C_0,\Gamma)$.
The probability that the execution reaches $\Tswap\cup \Stim$ within $O(n)$ time is at least \(1/20\).
\end{lemma}

\begin{proof}
First, we show that the expected time until either \textsc{Propagate-Reset} occurs or the execution reaches \( \Sdec \) is at most $1/2+(n-1)/2$.
We will derive this by analyzing the time to reach $\Sdec$ without \textsc{Propagate-Reset} occurring.
From Lemma~\ref{em:swap}, the execution reaches $\Sswap$ within $n$ expected time.
If \( n \) is even, for the execution starting from a configuration belonging to \( \Sswap \) to reach \( \Sdec \), it is sufficient for the agent with \( \Rank = n/2 \) and the agent with \( \Rank = n/2 + 1 \) to interact.
Since the probability that those two agents interact in each interaction is \( 2/(n(n-1)) \), the expected time for this to happen is \( (n-1)/2 \).  
If \( n \) is odd, for the execution starting from a configuration belonging to \( \Sswap \) to reach \( \Sdec \), it is sufficient for the agent with \( \Rank = \lceil n/2 \rceil \) to participate in an interaction.  
Since the probability of that agent participating in an interaction is \( 2/n \), the expected time is \( 1/2 \).
Thus, the execution starting from $C_0$ reaches $\Sdec$ within $1/2+(n-1)/2$ in expectation.

After the execution reaches \( \Sdec \), if there are agents with an \( \ans \) that is not the exact majority opinion, an interaction between the agent with $\Rank=\lceil n/2\rceil$ and such an agent will trigger \textsc{Propagate-Reset}.  
The probability of this happening is at least \( 2/(n(n-1)) \), so the expected time is \( (n-1)/2 \).
If the \( \ans \) of all agents are the exact majority opinion, the execution has reached to \( \Stim \).
Therefore, the expected time the execution starting from $C_0$ reaches $\Stim$ is bounded by $1/2+(n-1)/2+(n-1)/2=n-1/2$.
From Markov' inequality, the probability that the execution reaches $\Stim$ within $2n-1$ time is at least $1/2$.

Now, we analyze the time and probability for the execution to reach \( \Tswap \) after \textsc{Propagate-Reset} occurs.
From Lemma~\ref{ssrk:thm} and Markov's inequality, if the agent with \( \Rank = \lceil n/2 \rceil \) does not trigger \textsc{Propagate-Reset} again, the execution reaches \( \Srank \) within \( 2 \cdot s_{rank} \) time with probability at least \( 1/2 \).
During the ranking, when the agent with \( \Rank = \lceil n/2 \rceil \) is created, its $\timer$ is set to \( 7(t_{rank} + 4) \).  
Thus, we analyze the probability that its \( \timer \geq 28 \) when the execution reaches \( \Srank \).
Since \( s_{rank} \leq t_{rank} \cdot n \), we analyze the probability that \( \timer \geq 28 \) after \( n^2 \cdot t_{rank} \) interactions.
Let \( Z \sim B(n^2\cdot t_{rank}, 2/(n(n-1))) \) be a binomial random variable representing the number of times the agent with \( \Rank = \lceil n/2 \rceil \) interacts with the agent with $\Rank=n$ during \( n^2\cdot t_{rank} \) interactions.
From the Chernoff bound (Eq.~(4.2) in~\cite{Arisu}, with \( \delta = 2/3 \)),
$\Pr(Z\ge(1+2/3)E[Z])\le e^{-2/3(2/3)^2}<4/5$.
Since $(1+2/3)E[Z]\le \frac{20}{3}t_{rank}< 7t_{rank}$, that probability is at last $1/5$.
Therefore, after \textsc{Propagate-Reset} occurs, the execution reaches $\Tswap$ within $O(n)$ time with probability at least $1/10$. 

From the above, the execution reaches $\Tswap\cup \Stim$ within $O(n)$ time with probability at least $1/20$.
\end{proof}

\begin{lemma}\label{em:dec}
Let $C_0\in \Tswap$ and $\Xi_\Pem(C_0,\Gamma)$.
The probability that the execution reaches \( \Sdec \) within $3n-1$ time is at least \( 1/8 \).
\end{lemma}

\begin{proof}
We analyze the case where the execution reaches \( \Sdec \) without \textsc{Propagate-Reset} occurring.

First, we analyze the time and probability for the execution to reach \( \Sdec \) if \textsc{Propagate-Reset} does not occur.
From Lemma~\ref{em:swap} and Markov's inequality, agents finish swapping within $2n$ time with probability at least $1/2$.
After swapping, decision is finished within $(n-1)/2$ time in expectation as shown in the proof of Lemma~\ref{em:tswap}.
From Markov's inequality, agents finish decision within $n-1$ time with probability at least $1/2$.
Therefore, agents finish swapping and decision within $3n-1$ time with probability at least $1/4$.

Next, we analyze the probability that \textsc{Propagate-Reset} does not occur within \( 4n \) time.
Let $Z\sim B(4n^2, 2/(n(n-1)))$ be a binomial random variable representing the number of times the agent with $\Rank=\lceil n/2\rceil$ interacts with the agent with $\Rank=n$ during $4n^2$ interactions.
From the Chernoff bound (Eq.~(4.2) in~\cite{Arisu}, with \( \delta = 2/3 \)) and $E[Z]\ge 8$, $\Pr(Z\ge (1+2/3)E[Z])\le e^{-8/3(2/3)^2}\le 1/2$.
Since $(1+2/3)E[Z]\le 80/3<27$, the \( \timer \) of the agent with \( \Rank = \lceil n/2 \rceil \) does not reach $0$ during $4n$ time with probability at least $1/2$.

From the above, the probability for the execution to reach \( \Sdec \) within $3n-1$ time is at least \( 1/8 \).
\end{proof}

\begin{lemma}\label{em:timzero}
Let $C_0\in \Srank$ and $\Xi_\Pem(C_0,\Gamma)$.
The variable $\timer$ of the agent with $\Rank=\lceil n/2 \rceil$ reaches $0$ within $O(n)$ time with probability at least $1/2$.
\end{lemma}

\begin{proof}
Let \( x \) be the value of the \( \timer \) of the agent with \( \Rank = \lceil n/2 \rceil \).  
We analyze the expected time for the agent with \( \Rank = \lceil n/2 \rceil \) to interact with the agent with \( \Rank = n \) exactly \( x \) times.
Since the probability that those two agents interact is $2/(n(n-2))$, the expected time to interact once is $(n-1)/2$.
By the linearity of expectation, the expected time for the \( \timer \) to reach 0 is \( x(n-1)/2 \).
From Markov's inequality, the $\timer$ reaches $0$ within $x(n-1)\le 7(t_{rank}+4)(n-1)=O(n)$ time with probability at least $1/2$.
\end{proof}

\begin{lemma}\label{em:tim}
Let $C_0\in \Sdec$ and $\Xi_\Pem(C_0,\Gamma)$.
The probability that the execution reaches \( \Stim\) within \( O(n) \) time is at least \( 1/1280 \).
\end{lemma}

\begin{proof}
If there is no agent whose $\ans$ is not the exact majority opinion in $C_0$, then $C_0\in \Stim$.

Now, we consider the case where there are agents whose $\ans$ is not the exact majority opinion.
In this case, after the \( \timer \) of the agent with \( \Rank = \lceil n/2 \rceil \) reaches 0, when the agent interacts with an agent whose \( \ans \) is not the exact majority opinion, \textsc{Propagate-Reset} occurs, and the exact majority opinion propagates to all agents.  
The variable \( \timer \) of the agent with \( \Rank = \lceil n/2 \rceil \) reaches 0 within \( O(n) \) time with probability at least \( 1/2 \), from Lemma~\ref{em:timzero}.
Since the probability that agent with $\Rank=\lceil n/2\rceil$ interacts with the agent whose $\ans$ is not the exact majority opinion is at least $2/(n(n-1))$, the expected time is $(n-1)/2$.
From Markov's inequality, such an interaction occurs within \( n - 1 \) time with probability at least \( 1/2 \).
Thus, \textsc{Propagate-Reset} occurs within $O(n)$ time with probability at least $1/4$.

During \textsc{Propagate-Reset}, the exact majority opinion propagates to all agents within $O(\log{n})$ time with high probability from Lemma~\ref{ssrk:prop}.
Since \textsc{Propagate-Reset} occurs, agents also execute \textsc{Optimal-Silent-SSR}.
Note that \textsc{Propagate-Reset} finishes within $O(n)$ time with high probability, and \textsc{Optimal-Silent-SSR} finishes within $2\cdot O(n)$ time with probability at least $1/2$ from Lemma~\ref{ssrk:reset} and Markov's inequality.
After both protocols finished correctly, the $\timer$ of the agent with $\Rank=\lceil n/2\rceil$ is no less than $28$ with probability at least $1/10$ from the latter part of the proof of Lemma~\ref{em:tswap}.
Thus, the execution reaches $\Tswap \cap \Sout$ within $O(n)$ time with probability at least $1/4\cdot 1/10\cdot (1-1/n)^2\ge 1/160$.

From Lemma~\ref{em:dec}, the execution belonged $\Tswap$ reaches $\Sdec$ within $3n-1$ time with probability $1/8$.
Thus, since there is no agent whose $\ans$ is not the exact majority opinion, the execution reaches $\Sdec \cap \Sout\subseteq \Stim$.
Therefore, the probability that the execution reaches $\Stim$ within $O(n)$ time is at least $1/160 \cdot 1/8=1/1280$.
\end{proof}

\begin{lemma}\label{em:safe}
Let $C_0\in \Stim$ and $\Xi_\Pem(C_0,\Gamma)$.
The probability that the execution reaches \( \Sem\) within \( O(n) \) time is at least \( 1/2 \).
\end{lemma}

\begin{proof}
For the execution to reach \( \Sem \), it is necessary for the \( \timer \) of the agent with \( \Rank = \lceil n/2 \rceil \) to reach $0$.
From Lemma~\ref{em:timzero}, the \( \timer \) of the agent with \( \Rank = \lceil n/2 \rceil \) reaches $0$ within $O(n)$ time with probability at least $1/2$.
Thus, the execution reaches $\Sem$ within $O(n)$ time with probability at least $1/2$.
\end{proof}

\begin{lemma}\label{em:stab}
$\Pem$ reaches $\Sem$ with probability $1$, and the silence time of $\Pem$ is \( O(n) \) in expectation, and $O(n\log{n})$ with high probability.
\end{lemma}

\begin{proof}
Let $C_0\in \Call(\Pem)$ and $\Xi_\Pem(C_0,\Gamma)$.
From Lemma~\ref{em:rank},~\ref{em:tswap},~\ref{em:dec},~\ref{em:tim}, and~\ref{em:safe}, the execution reaches $\Sem$ with probability at least $1/4096000$.
Since an execution reaches \( \Sem \) with nonzero probability, the execution eventually reaches \( \Sem \) if it runs long enough.  
Thus, there exists an execution that reaches a silent configuration from any initial configuration; therefore, \( \Pem \) reaches \( \Sem \) with probability 1.

Let $E_{\Pem}$ be the expected time for the execution to reach $\Sem$.
From the above, $E_{\Pem}\le O(n)+(1-1/4096000)E_{\Pem}$ holds.
Thus, $E_{\Pem}=O(n)$.
From Markov's inequality, the execution reaches $\Sem$ within $2\cdot O(n)$ time with probability at least $1/2$.
Thus, the execution reaches $\Sem$ within $2\cdot O(n)\log{n}=O(n\log{n})$ time with probability at least $1-(1/2)^{\log{n}}=1-1/n$.
\end{proof}

From Theorem~\ref{lb:states}, ~\ref{lb:time}, and Lemma~\ref{em:stab}, we derive the following theorem.

\begin{theorem}
\( \Pem \) is a time- and space-optimal silent self-stabilizing exact majority protocol that uses \( O(n) \) states and reaches a silent configuration within \( O(n) \) time in expectation and within \( O(n \log n) \) time with high probability.
\end{theorem}

\section{Conclusion and Discussions}\label{sec:conclude}
We addressed a silent protocol that solves the self-stabilizing exact majority problem.  
We showed that no protocol can solve this problem without knowledge of \( n \).  
We proposed a silent protocol within \( O(n) \) time in expectation and \( O(n \log n) \) time with high probability, using $O(n)$ states, with knowledge of \( n \).
We established lower bounds, proving that any protocol requires \( \Omega(n) \) states, \( \Omega(n) \) time in expectation, and \(\Omega(n\log{n})\) time with high probability to reach a safe configuration.  
Thus, the proposed protocol is time- and space-optimal.

We propose the following open problems:
An extension of the self-stabilizing majority problem to the exact plurality consensus problem~\cite{BBBEHKK22}, in which agents have $k$ different opinions, can also be considered.
Whether this problem can achieve the lower bound shown in this paper remains an intriguing open question.
Additionally, it is worth considering a non-silent self-stabilizing exact majority protocol with fewer states.

\paragraph*{Acknowledgments.}
This work was supported by JSPS KAKENHI Grant Numbers JP22H03569, JP23K28037, and JP25K03101.

\bibliographystyle{plain}
\bibliography{ref}

\end{document}